\ifx\SoCG\undefined%
\def\UseBibLatex{1}
\documentclass[12pt]{article}%
\providecommand{\SoCGVer}[1]{}%
\providecommand{\NotSoCGVer}[1]{#1}%
\else%
\makeatletter
\def\input@path{{lipics/}{../lipics/}}
\makeatother
\documentclass[a4paper,USenglish,cleveref,autoref,thm-restate]%
{lipics-v2021}

\hideLIPIcs
\providecommand{\SoCGVer}[1]{#1}%
\providecommand{\NotSoCGVer}[1]{}%
\fi

\ifx\sarielComp\undefined%
    \newcommand{\SarielComp}[1]{} \newcommand{\NotSarielComp}[1]{#1}%
\else
    \newcommand{\SarielComp}[1]{#1}%
    \newcommand{\NotSarielComp}[1]{}%
\fi \newcommand{\IfPrinterVer}[2]{#2}%

\newcommand{\UsePackage}[1]{%
  \IfFileExists{styles/#1.sty}{%
      \usepackage{styles/#1}%
   }{%
      \IfFileExists{../styles/#1.sty}{%
         \usepackage{../styles/#1}%
      }{%
         \usepackage{#1}%
      }%
   }%
}

\NotSoCGVer{%
	\usepackage[in]{fullpage}%
}

\usepackage{amsmath}%
\usepackage{amssymb}%
\usepackage{bbm}%
\usepackage{caption}%
\usepackage{xcolor}%
\usepackage{mathtools}%
\usepackage{commath}

\usepackage{scalerel}%

\SarielComp{\usepackage{sariel_colors}}%

\NotSoCGVer{%
	\usepackage[amsmath,thmmarks]{ntheorem}%
	\theoremseparator{.}%
}%

\usepackage{titlesec}%
\titlelabel{\thetitle. }%

\titleformat{\paragraph}[runin]
{\normalfont\bfseries}
{\theparagraph}
{1em}
{\addperiod}

\newcommand{\addperiod}[1]{#1.}

\makeatletter
\patchcmd{\ttlh@hang}{\parindent\z@}{\parindent\z@\leavevmode}{}{}
\patchcmd{\ttlh@hang}{\noindent}{}{}{}
\makeatother

\usepackage{bbm}
\usepackage{titlesec}%
\titlelabel{\thetitle. }%
\usepackage{xcolor}%
\usepackage{mleftright}%
\usepackage{xspace}%
\usepackage{hyperref}%

\IfPrinterVer{%
   \usepackage{hyperref}%
}{%
   \usepackage{hyperref}%
   \hypersetup{%
      breaklinks,%
      colorlinks=true,%
      urlcolor=[rgb]{0.25,0.0,0.0},%
      linkcolor=[rgb]{0.5,0.0,0.0},%
      citecolor=[rgb]{0,0.2,0.445},%
      filecolor=[rgb]{0,0,0.4},
      anchorcolor=[rgb]={0.0,0.1,0.2}%
   }
   \usepackage[ocgcolorlinks]{ocgx2}
}

\usepackage{textgreek}

\providecommand{\BibLatexMode}[1]{}
\providecommand{\BibTexMode}[1]{#1}

\ifx\UseBibLatex\undefined%
  \renewcommand{\BibLatexMode}[1]{}
  \renewcommand{\BibTexMode}[1]{#1}
\else
  \renewcommand{\BibLatexMode}[1]{#1}
  \renewcommand{\BibTexMode}[1]{}
\fi

\BibLatexMode{%
   \usepackage[bibencoding=utf8,style=alphabetic,backend=biber]{biblatex}%
   \UsePackage{sariel_biblatex}%
}

\NotSoCGVer{%

\theoremstyle{plain}%
\newtheorem{theorem}{Theorem}[section]

\newtheorem{lemma}[theorem]{Lemma}

\newtheorem{corollary}[theorem]{Corollary}

\theoremstyle{plain}%
\theoremheaderfont{\sf} \theorembodyfont{\upshape}%
\newtheorem*{remark:unnumbered}[theorem]{Remark}%
\newtheorem{remark}[theorem]{Remark}%
\newtheorem{definition}[theorem]{Definition}

\newcommand{\myqedsymbol}{\rule{2mm}{2mm}}

\theoremheaderfont{\em}%
\theorembodyfont{\upshape}%
\theoremstyle{nonumberplain}%
\theoremseparator{}%
\theoremsymbol{\myqedsymbol}%
\newtheorem{proof}{Proof:}%

}

\definecolor{blue25emph}{rgb}{0, 0, 11}
\providecommand{\emphic}[2]{%
   \textcolor{blue25emph}{%
      \textbf{\emph{#1}}}%
   \index{#2}}

\providecommand{\emphi}[1]{\emphic{#1}{#1}}

\definecolor{almostblack}{rgb}{0, 0, 0.3}
\providecommand{\emphw}[1]{{\textcolor{almostblack}{\emph{#1}}}}%

\newcommand{\atgen}{\symbol{'100}} \newcommand{\EliotThanks}[1]{%
   \thanks{%
      Department of Computer Science; University of Illinois; 201
      N. Goodwin Avenue; Urbana, IL, 61801, USA; \\%
      {\tt erobson2\atgen{}illinois.edu}; {\tt
         \url{https://eliotwrobson.github.io/}.} #1}}
\newcommand{\SarielThanks}[1]{\thanks{Department of Computer Science;
      University of Illinois; 201 N. Goodwin Avenue; Urbana, IL,
      61801, USA; \\%
      {\tt sariel\atgen{}illinois.edu}; {\tt
         \url{http://sarielhp.org/}.} #1}}

\newcommand{\etal}{\textit{et~al.}\xspace}
\newcommand{\HLink}[2]{\hyperref[#2]{#1~\ref*{#2}}}
\newcommand{\HLinkSuffix}[3]{\hyperref[#2]{#1\ref*{#2}{#3}}}

\newcommand{\thmlab}[1]{{\label{theo:#1}}}
\renewcommand{\thmref}[1]{\HLink{Theorem}{theo:#1}}

\newcommand{\remlab}[1]{\label{rem:#1}}

\newcommand{\defrefY}[2]{\hyperref[def:#1]{#2}}

\newcommand{\apndlab}[1]{\label{apnd:#1}}
\newcommand{\apndref}[1]{\HLink{Appendix}{apnd:#1}}

\newcommand{\lemlab}[1]{\label{lemma:#1}}
\renewcommand{\lemref}[1]{\HLink{Lemma}{lemma:#1}}%

\newcommand{\seclab}[1]{\label{section:#1}}
\renewcommand{\secref}[1]{\HLink{Section}{section:#1}}%

\providecommand{\eqlab}[1]{}%
\renewcommand{\eqlab}[1]{\label{equation:#1}}
\newcommand{\Eqref}[1]{\HLinkSuffix{Eq.~(}{equation:#1}{)}}

\providecommand{\remove}[1]{}%
\newcommand{\Set}[2]{\left\{ #1 \;\middle\vert\; #2 \right\}}
\newcommand{\pth}[2][\!]{\mleft({#2}\mright)}%
\newcommand{\pbrcx}[1]{\left[ {#1} \right]}%
\newcommand{\ProbLTR}{\mathbb{P}}%
\newcommand{\Prob}[1]{\mathop{\ProbLTR} \mleft[ #1 \mright]}%

\newcommand{\Ex}[2][\!]{\mathop{\mathbb{E}}#1\pbrcx{#2}}
\newcommand{\ExCond}[2]{\ExSym\!\left[%
       #1 \;\middle\vert\; #2 \right]}
\newcommand{\ExChar}{\mathbb{E}}%
\newcommand{\ExSym}{\mathop{\ExChar}}%
\newcommand{\Var}[1]{\mathop{\mathbb{V}}\!\pbrcx{#1}}

\newcommand{\ceil}[1]{\left\lceil {#1} \right\rceil}
\newcommand{\floor}[1]{\left\lfloor {#1} \right\rfloor}

\newcommand{\cardin}[1]{\left| {#1} \right|}%

\renewcommand{\th}{th\xspace}

\renewcommand{\Re}{\mathbb{R}}%
\newcommand{\eps}{\varepsilon}

\usepackage[inline]{enumitem}

\newlist{compactenumA}{enumerate}{5}%
\setlist[compactenumA]{topsep=0pt,itemsep=-1ex,partopsep=1ex,parsep=1ex,%
   label=(\Alph*)}%

\newlist{compactenuma}{enumerate}{5}%
\setlist[compactenuma]{topsep=0pt,itemsep=-1ex,partopsep=1ex,parsep=1ex,%
   label=(\alph*)}%

\newlist{compactenumI}{enumerate}{5}%
\setlist[compactenumI]{topsep=0pt,itemsep=-1ex,partopsep=1ex,parsep=1ex,%
   label=(\Roman*)}%

\newlist{compactenumi}{enumerate}{5}%
\setlist[compactenumi]{topsep=0pt,itemsep=-1ex,partopsep=1ex,parsep=1ex,%
   label=(\roman*)}%

\newlist{compactitem}{itemize}{5}%
\setlist[compactitem]{topsep=0pt,itemsep=-1ex,partopsep=1ex,parsep=1ex,%
   label=\bullet}%

\DeclareFontFamily{U}{BOONDOX-calo}{\skewchar\font=45 }
\DeclareFontShape{U}{BOONDOX-calo}{m}{n}{
  <-> s*[1.05] BOONDOX-r-calo}{}
\DeclareFontShape{U}{BOONDOX-calo}{b}{n}{
  <-> s*[1.05] BOONDOX-b-calo}{}
\DeclareMathAlphabet{\mathcalb}{U}{BOONDOX-calo}{m}{n}
\SetMathAlphabet{\mathcalb}{bold}{U}{BOONDOX-calo}{b}{n}
\DeclareMathAlphabet{\mathbcalb}{U}{BOONDOX-calo}{b}{n}

\usepackage{stmaryrd}%
\providecommand{\IntRange}[1]{\mleft\llbracket #1 \mright\rrbracket}
\newcommand{\IRX}[1]{\IntRange{#1}}%

\providecommand{\Mh}[1]{#1}%

\newcommand{\Caratheodory}{Carath\'eodory\xspace}

\renewcommand{\P}{\Mh{P}}%
\renewcommand{\L}{\Mh{L}}%
\newcommand{\Q}{\Mh{Q}}%

\newcommand{\CHX}[1]{\Mh{\mathsf{CH}}\pth{#1}}

\newcommand{\cenX}[1]{\Mh{\overline{\mathsf{c}}}_{#1}}
\newcommand{\centroX}[1]{\Mh{\overline{\mathsf{m}}}_{#1}}
\newcommand{\permut}[1]{\left\langle {#1} \right\rangle}
\newcommand{\DotProdY}[2]{\permut{{#1},{#2}}}

\newcommand{\normX}[1]{\left\| {#1} \right\|}%

\newcommand{\ball}{\Mh{\mathcalb{b}}}%
\newcommand{\ballY}[2]{\ball\pth{#1, #2}}%

\numberwithin{figure}{section}%
\numberwithin{table}{section}%
\numberwithin{equation}{section}%

\DeclareMathAlphabet{\mathpzc}{OT1}{pzc}{m}{it}

\newcommand{\ee}{\Mh{\mathcalb{e}}}

\newcommand{\N}{\Mh{\mathsf{N}}}%

\newcommand{\Otild}{\Mh{\widetilde{O}}}

\newcommand{\SaveContent}[2]{%
   \expandafter\newcommand{#1}{#2}%
}

\newcommand{\RestatementOf}[2]{
   \noindent%
   \textbf{Restatement of #1.}
   {\em #2{}}%
}

\newcommand{\Event}{\Mh{\mathcal{E}}}%
\newcommand{\avgpX}[1]{\Mh{\nabla}\pth{#1}}%
\newcommand{\avgp}{\Mh{\nabla}}%
\newcommand{\Sample}{\Mh{R}}%

\providecommand{\TPDF}[2]{\texorpdfstring{#1}{#2}}

\providecommand{\Mh}[1]{#1}%

\newcommand{\cp}{\Mh{\mathsf{c}}}%
\newcommand{\cDelta}{c}
\newcommand{\cpNR}{\Mh{\mathsf{c}_{n,r}}}
\newcommand{\dY}[2]{\left\| #1 #2 \right\|}

\newcommand{\radiusX}[1]{\mathrm{radius}\pth{#1}}
\newcommand{\ts}{\hspace{0.6pt}}

\renewcommand{\P}{\Mh{P}}
\newcommand{\q}{\Mh{q}}

\newcommand{\p}{\Mh{p}}%

\newcommand{\epsA}{\Mh{\delta}}%

\newcommand{\espB}{\xi}

\newcommand{\diam}{\Mh{\Delta}}
\newcommand{\diamX}[1]{\mathrm{diam}\pth{#1}}

\newcommand{\TSet}{\Mh{\mathcal{H}}}

\BibLatexMode{%
   \bibliography{fat_separator} }

\begin{document}
\title{No-dimensional Tverberg Partitions Revisited%
   \footnote{A preliminary version of this paper appeared in SWAT
      2024. A full version is in submission to DCG.}%
}

\author{%
   Sariel Har-Peled%
   \SarielThanks{Work on this paper was partially supported by NSF AF
      award CCF-2317241.%
   }%
   \and
   Eliot W. Robson%
   \EliotThanks{Work on this paper was partially supported by NSF
      AF awards CCF-2317241.}%
}%
\date{\today}
\maketitle

\begin{abstract}
   Given a set $\P \subset \Re^d$ of $n$ points, with diameter
   $\diam$, and a parameter $\epsA \in (0,1)$, it is known that there
   is a partition of $\P$ into sets $\P_1, \ldots, \P_t$, each of size
   $O(1/\epsA^2)$, such that their convex hulls all intersect a common
   ball of radius $\epsA \diam$. We prove that a random partition,
   with a simple alteration step, yields the desired partition,
   resulting in a (randomized) linear time algorithm (i.e.,
   $O(dn)$). We also provide a deterministic algorithm with running
   time $O( dn \log n)$. Previous proofs were either existential
   (i.e., at least exponential time), or required much bigger sets.
   In addition, the algorithm and its proof of correctness are
   significantly simpler than previous work, and the constants are
   slightly better.

   We also include a number of applications and extensions using the
   same central ideas. For example, we provide a linear time algorithm
   for computing a ``fuzzy'' centerpoint, and prove a no-dimensional
   weak $\eps$-net theorem with an improved constant.
\end{abstract}

\section{Introduction}

\paragraph*{Centerpoints}

A point $\cp$ is an \emphw{$\alpha$-centerpoint} of a set
$\P \subseteq \Re^d$ of $n$ points, if all closed halfspaces
containing $\cp$ also contain at least $\alpha n$ points of $\P$.  The
parameter $\alpha$ is the \emphw{centrality} of $\cp$, while
$\alpha n$ is its \emphi{Tukey depth}. The centerpoint theorem \cite{m-ldg-02},
which is a consequence of Helly's theorem \cite{h-usamp-30}, states that a
$1/(d+1)$-centerpoint (denoted $\cenX{\P}$) always exists.

In two dimensions, Jadhav and Mukhopadhyay~\cite{jm-ccfpsplt-94}
presented an $O(n)$ time algorithm for computing a $1/3$-centerpoint
(but not the point of maximum Tukey depth). Chan \etal
\cite{chj-oagcd-22} presented an $O(n \log n+ n^{d-1})$ algorithm for
computing the point of maximum Tukey depth (and thus also a
$1/(d+1)$-centerpoint). It is believed that $\Omega(n^{d-1})$ is a
lower bound on solving this problem exactly, see \cite{chj-oagcd-22}
for details and history.

This guarantee of a $1/(d+1)$-centerpoint is tight, as demonstrated by
placing the points of $\P$ in $d+1$ small, equal size clusters
(mimicking weighted points) in the vicinity of the vertices of a
simplex. Furthermore, the lower-bound of $\ceil{n/ (d+1)}$ is all but
meaningless if $d$ is as large as $n-1$.

\paragraph*{Approximating centrality}
A randomized $\Otild(d^9)$ time algorithm for computing a (roughly)
$1/(4d^2)$ centerpoint was presented by Clarkson \etal
\cite{cemst-acpir-96}, and a later refinement by Har-Peled and Jones
\cite{hj-jcps-21} improved this algorithm to compute a (roughly)
$1/d^2$ centerpoint in $\Otild(d^7)$ time, where $\Otild$ hides
polylog terms. Miller and Sheehy \cite{ms-acp-10} derandomized the
algorithm of Clarkson \etal, computing a $\Omega(1/d^2)$ centerpoint
in time $n^{O(\log d)}$. This result was later applied by Mulzer and
Werner \cite{mw-atplt-13}, computing a $\Omega(1/d^3)$ centerpoint in
time $d^{O(\log d)} n$. Developing an algorithm that computes a
$1/(d+1)$-centerpoint in polynomial time (in $d$) in still open,
although the existence of such an algorithm with running time better
than $\Omega(n^{d-1})$ seems unlikely, as mentioned above.

\paragraph*{Tverberg partitions}
Consider a set $\P$ of $n$ points in $\Re^d$. Tverberg's theorem
states that such a set can be partitioned into $k = \floor{n/(d+1)}$
subsets, such that all of their convex-hulls intersect. Specifically,
a point in this common intersection is a $1/(d+1)$-centerpoint.
Indeed, a point $\p$ contained in the convex-hulls of the $k$ sets of
the partition is a $k/n$-centerpoint, as any halfspace containing $\p$
must also contain at least one point from each of these $k$ subsets.
Refer to the surveys \cite{dgmm-dyutc-19} and \cite{bs-ttfyo-18} for
information on this and related theorems.

This theorem has an algorithmic proof \cite{tv-grths-93}, but its
running time is $n^{O(d^2)}$.  To understand the challenge in getting
an efficient algorithm for this problem, observe that it is not known,
in strongly polynomial time, to decide if a point is inside the
convex-hull of a point set (i.e., is it a
$1/n$-centerpoint?). Similarly, for a given point $\p$, it is not
known how to compute, in weakly or strongly polynomial time, the
centrality of $\p$.  Nevertheless, a Tverberg partition is quite
attractive, as the partition itself (and its size) provides a compact
proof (i.e., lower bound) of its centrality. If the convex-combination
realization of $\p$ inside each of these sets is given, then its
$k/n$-centrality can be verified in linear time.

There has been significant work trying to compute Tverberg partitions
with as many sets as possible while keeping the running time
polynomial. The best polynomial algorithms currently known (roughly)
match the bounds for the approximate centerpoint mentioned
above. Specifically, it is known how to compute a Tverberg partition
of size $O\bigl(n/(d^2 \log d)\bigr)$ (along with a point in the
common intersection) in weakly polynomial time. See \cite{tv-grths-93,
   ms-acp-10, mw-atplt-13, rs-aatt-16, mmss-relpf-17,
   cm-ndtta-22,hz-iaatp-21} and references therein.

\paragraph*{No-dimensional Tverberg theorem}
Adiprasito \etal \cite{abmt-tchtw-20} proved a no-dimensional variant
of Tverberg's theorem. Specifically, for $\epsA \in (0,1)$, they
showed that one can partition a point set $\P$ into sets of size
$O(1/\epsA^2)$, such that the convex-hulls of the sets intersect a
common ball of radius $\epsA\ts \diamX{\P}$. Their result is
existential and does not yield an efficient algorithm. However, as the
name suggests, it has the attractive feature that the sets in the
partition have size that does not depend on the dimension.

Choudhary and Mulzer \cite{cm-ndtta-22} gave a weaker version of this
theorem, but with an efficient algorithm. Informally, given a set
$\P \subset \Re^d$ of $n$ points, and a parameter $\epsA \in (0,1)$,
$\P$ can be partitioned, in $O( n d \log k)$ time, into
$k = O(\epsA \sqrt{n})$ sets $\P_1, \ldots, \P_k$, each of size
$\Theta(\sqrt{n}/\epsA)$, such that there is a ball of radius
$\epsA\ts \diamX{\P}$ that intersects the convex-hull of $\P_i$ for
every $i$. Note that the later (algorithmic) result is significantly
weaker than the previous (existential) result, as the subsets have to
be substantially larger.

Thus, the question remains: Can one compute a no-dimensional Tverberg
partition with the parameters of Adiprasito \etal \cite{abmt-tchtw-20}
in linear time?

\paragraph*{Centerball via Tverberg partition}

As observed by Adiprasito \etal \cite{abmt-tchtw-20}, a no-dimensional
Tverberg partition readily implies a no-dimensional centerpoint
result, where the central point is replaced by a ball. Specifically,
they showed that one can compute a ball of radius
$\epsA \ts \diamX{\P}$ such that any halfspace containing it contains
$\Omega( \epsA^2 n)$ points of $\P$.

\paragraph*{Centroid and sampling}

The \emphi{centroid} of a point set $\P$ is the point
$\centroX{\P} = \sum_{\p \in \P} \p/\cardin{\P}$. The \emphi{$1$-mean}
price of clustering $\P$, using $\q$, is the sum of squared distances
of the points of $\P$ to $\q$, that is
$f(\q) = \sum_{\p \in \P} \dY{\p}{\q}^2$.  It is not hard to verify
that $f$ is minimized at the centroid $\centroX{\P}$.  A classical
observation of Inaba \etal \cite{iki-awvdr-94} is that a sample
$\Sample$ of size $O(1/\epsA^2)$ of points from $\P$ is $\epsA$-close
to the global centroid of the point set. That is,
$\dY{\centroX{\P}}{\centroX{\Sample}} \leq \epsA\ts \diamX{\P}$ with
constant probability.  Applications of this observation to $k$-means
clustering and sparse coresets are well known, see Clarkson
\cite[Section~2.4]{c-csgafwa-08} and references therein.

\paragraph*{Our results}

We show that the aforementioned observation of Inaba \etal implies the
no-dimensional Tverberg partition. Informally, for a random partition
of $\P$ into sets of size $O(1/\epsA^2)$, most of the sets are in
distance at most $\epsA \ts \diamX{\P}$ from the global centroid of
$\P$. By folding the far sets (i.e., ``bad''), into the close sets
(i.e., ``good''), we obtain the desired partition. The resulting
algorithm has (expected) linear running time $O(d n)$.

For the sake of completeness, we prove the specific form of the
$1$-mean sampling observation \cite{iki-awvdr-94} we need in
\lemref{sample:mean} -- the proof requires slightly tedious but
straightforward calculations. The linear time algorithm for computing
the no-dimensional Tverberg partition is presented in \thmref{main:a},
which is the main result of this paper.

In the other extreme, one wants to split the point set into two sets
of equal size while minimizing their distance.  We show that a set
$\P$ with $2n$ points can be split (in linear time) into two sets of
size $n$, such that (informally) the expected distance of their
centroids is $\leq \diamX{\P}/\sqrt{n}$. The proof of this is even
simpler (!), and the bound is tight; see \lemref{sample:mean:2}.  We
present several applications: %
\medskip%
\begin{compactenumI}
    \item \textsc{No-dimensional Centerball.}  In
    \secref{center:ball}, we present a no-dimensional generalization
    of the centerpoint theorem. As mentioned above, this was already
    observed by Adiprasito \etal \cite{abmt-tchtw-20}, but our version
    can be computed efficiently.

    \smallskip%
    \item \textsc{Weak $\eps$-net.} A new proof of the no-dimensional
    version of the weak $\eps$-net theorem with improved constants,
    see \secref{weak}.

    \smallskip%
    \item \textsc{Derandomization.}  The sampling mean lemma (i.e.,
    \lemref{sample:mean}) can be derandomized to yield a linear time
    algorithm, see \lemref{s:m:d:fast}. The somewhat slower version,
    \lemref{s:m:d:slow}, is a nice example of using conditional
    expectations for derandomization.  Similarly, the halving scheme
    of \lemref{sample:mean:2} can be derandomized in a fashion similar
    to discrepancy algorithms \cite{m-gd-99, c-dmrc-01}. The
    derandomized algorithm, presented in \lemref{sample:mean:3}, has
    linear running time $O(dn)$.

    This leads to a deterministic $O(dn \log n)$ time algorithm for
    the no-dimensional Tverberg partition, see \lemref{main:b}. The
    idea is to repeatedly apply the halving scheme, in a binary tree
    fashion, till the point set is partitioned into subsets of size
    $O(1/\epsA^2)$. Both the running time and constants are somewhat
    worse than the randomized algorithm of \thmref{main:a}, but it is
    conceptually even simpler, avoiding the need for an alteration
    step.
\end{compactenumI}
\medskip%
As an extra, another neat implication of the observation of Inaba
\etal \cite{iki-awvdr-94} is the dimension free version of
\Caratheodory's theorem \cite{m-ldg-02}, which we present for
completeness in \apndref{caratheodory}.

\paragraph*{Simplicity}

While simplicity is in the eyes of the beholder, the authors find the
brevity of the results here striking compared to previous work. In
particular, our presentation here is longer than strictly necessary,
as we reproduce proofs of previous known results, such as
\lemref{sample:mean} and its variant \lemref{sample:mean:2}, so our
work is self contained.

\section{Approximate Tverberg partition via mean sampling}

In the following, for two points $\p,\q \in \Re^d$, let
$\p\q = \DotProdY{\p}{\q} = \sum_{i=1}^d \p[i]\q[i]$ denote their
dot-product. Thus, $\p^2 = \DotProdY{\p}{\p} = \normX{\p}^2$.  Let
$\P$ be a finite set of points in $\Re^d$ (but any metric space equipped with a
dot-product suffices), and let
$\centroX{\P} = \sum_{\p \in \P} \p/\cardin{\P}$ denote the
\emphi{centroid} of $\P$. The \emphi{average price} of the $1$-mean
clustering of $\P$ is
\begin{equation}
    \avgpX{\P}
    =
    \sqrt{\sum\nolimits_{\p \in \P} \dY{\p}{\centroX{\P}}^2/\cardin{\P}}
    \leq
    \diamX{\P}.
    \eqlab{yo}
\end{equation}
The last inequality follows as $\centroX{\P} \in \CHX{\P}$, and for
any $\p \in \P$, we have $\dY{\p }{\centroX{\P}} \leq
\diamX{\P}$. This inequality can be tightened.

\begin{lemma}
    \lemlab{side:show}%
    We have $\avgpX{\P}\leq \diamX{\P}/\sqrt{2}$, and there is a point
    set $\Q$ in $\Re^d$, such that
    \begin{equation*}
        \avgpX{\Q}\geq \sqrt{1-\tfrac{1}{d}} \tfrac{1}{\sqrt{2}}
        \diamX{\Q}
    \end{equation*}
    (i.e., the inequality is essentially tight).
\end{lemma}
\begin{proof}
    This claim only improves the constant in our main result, and the
    reader can safely skip reading the proof.  Let $\P$ be a set of
    $n$ points in $\Re^d$, with $\diam = \diamX{\P}$ and
    $\avgp = \avgpX{\P}$.  Assume that $\centroX{\P} = 0$, as the
    claim is translation invariant. That is $\sum_{\q \in \P} \q =0$,
    and
    \begin{equation*}
        \beta =
        \sum_{\p,\q \in \P}
        \DotProdY{\p}{ \q}
        =%
        \sum_{\p \in \P}
        \DotProdY{\Bigl.\p}{\smash{\sum\nolimits_{\q \in \P}} \q}
        =%
        \sum_{\p \in \P}
        \DotProdY{\p}{0}
        =%
        0.
    \end{equation*}
    We have
    \begin{align*}
      n \avgp^2
      &=
        \sum_{\p \in \P} \normX{\p }^2
        =%
        \sum_{\p , \q \in \P} \frac{\normX{\p }^2+\normX{\q}^2}{2n}
        =%
        \sum_{\p , \q \in \P} \frac{\normX{\p }^2 -
        2\DotProdY{\p}{\q} + \normX{\q}^2}{2n}
        + \frac{2\beta}{2n}
        =%
        \sum_{\p , \q \in \P} \frac{\dY{\p }{\q}^2}{2n}
      \\&%
      \leq
      \sum_{\p \in \P, \q \in \P} \frac{\diam^2}{2n}
      =
      \frac{n^2 \diam^2}{2n}.
    \end{align*}
    Implying that $\avgp^2 \leq \diam^2/2$.

    As for the lower bound, let $\ee_i$ be the $i$\th standard unit
    vector\footnote{That is, $\ee_i$ is $0$ in all coordinates except
       the $i$\th coordinate where it is $1$.}  in $\Re^d$, and
    consider the point set $\Q = \{\ee_1, \ldots, \ee_d\}$. We have
    that $\diamX{\Q} = \sqrt{2}$ and
    $\centroX{\Q} = (1/d,\ldots, 1/d)$. Consequently,
    \begin{align*}
      \avgpX{\Q}
      &=
        \sqrt{\tfrac{1}{\cardin{\Q}}\sum\nolimits_{\q \in \Q} \dY{\q
        }{\centroX{\Q}}^2}
        =
        \sqrt{ \tfrac{\cardin{\Q}}{\cardin{\Q}}
        \pth{ (1-1/d)^2 + (d-1)/d^2}}
      =%
      \sqrt{   \frac{(d-1)^2 + d-1}{d^2}}
        \\&
      =
      \sqrt{\frac{d-1}{d}}
      =
      \frac{\diamX{\Q}}{\sqrt{2}}
      \sqrt{1 - \frac{1}{d}}%
      \geq%
      \pth{1-\frac{1}{d}}\frac{1}{\sqrt{2}}\diamX{\Q}.
    \end{align*}
\end{proof}

\begin{definition}
    A subset $X \subseteq \P$ is \emphi{$\epsA$-close} if the centroid
    of $X$ is in distance at most $\epsA \ts \diamX{\P}$ from the
    centroid of $\P$ -- that is,
    $\dY{\centroX{X}}{\centroX{\P}} \leq \epsA \ts \diamX{\P}$.
\end{definition}

\subsection{Proximity of centroid of a sample}

The following is by now standard -- a random sample of $O(1/\epsA^2)$
points from $\P$ is $\epsA$-close with good probability, see Inaba
\etal \cite[Lemma 1]{iki-awvdr-94}.  We include the proof for the sake
of completeness, as we require this somewhat specific form.

\begin{lemma}
    \lemlab{sample:mean}%
    Let $\P$ be a set of $n$ points in $\Re^d$, and $\epsA \in (0,1)$
    be a parameter. Let $R \subseteq \P$ be a random sample of size
    $r$ picked uniformly without replacement from $\P$, where
    $r \geq \zeta /\epsA^2$ and $\zeta > 1$ is a parameter.  Then, we
    have
    $\Prob{\dY{\centroX{\P}}{\centroX{R}} > \epsA \avgpX{\P}} <
    1/\zeta$.
\end{lemma}

\begin{proof}
    Let $\P=\{\p_1,\ldots, \p_n\}$. For simplicity of exposition,
    assume that $\centroX{\P} = \sum_{i=1}^n \tfrac{1}{n} \p_i=0$, as
    the claim is translation invariant. For $\avgp = \avgpX{\P}$,
    we have $\avgp^2 = \sum_{i=1}^n \tfrac{1}{n}\p_i^2$.  Let
    $\Bigl.Y = \sum_{\p \in R} \p = \sum_{i=1}^n I_i \p_i$, where
    $I_i$ is an indicator variable for $\p_i$ being in $R$. By
    linearity of expectations, we have
    \begin{equation*}
    	\Ex{Y}%
        =%
        \sum_{i=1}^n \Ex{I_i} \p_i%
        =%
        \sum_{i=1}^n \frac{r}{n} \p_i%
        =%
        r \sum_{i=1}^n \frac{1}{n} \p_i%
        =%
        r \, \centroX{\P}
        =
        0.
    \end{equation*}
    Observe that, for $i \neq j$, we define
    \begin{align}
      \cpNR =
      \Ex{I_i {}I_j}%
      =%
      \Prob{I_i= 1 \text{ and }I_j = 1}%
      =%
      \frac{{\textstyle\binom{n-2}{r-2}}}{\textstyle{\binom{n}{r}}}%
      \eqlab{pair:ex}
      =%
      \frac{r(r-1)}{n(n-1)}.
    \end{align}
    which notably does not depend on $i$ or $j$.  By the above, and
    since $\Ex{\smash{I_i^2}\bigr.}=\Ex{I_i}$, we have
    \begin{align}
      \nonumber%
      \Ex{\bigl.\smash{\normX{Y }^2}}
      &=
        \Ex{\bigl. \DotProdY{Y}{Y}}%
        =%
        \Ex{\Bigl.\smash{\bigl(\sum_{i=1}^n I_i p_i \bigr)^2}}
        =%
        \Bigl.\smash{\sum_{i=1}^n \Ex{I_i} \p_i^2 +
        2 \sum_{i<j} \Ex{I_i {}I_j} \p_i\p_j
        }
      \\[0.3cm]
      &
        =%
        \sum_{i=1}^n \frac{r}{n} \p_i^2 +
        2 \cpNR \sum_{i<j} \p_i\p_j%
        \leq%
        r \avgp^2
        + n^2 \cpNR
        \Bigl(\sum_{i=1}^n \frac{1}{n} \p_i \Bigr)^{\!2}%
        =%
        r \avgp^2,
        \eqlab{y:squared}
    \end{align}
    using the shorthand $\p_i \p_j = \DotProdY{\p_i}{\p_j}$ and
    $\p_i^2 = \DotProdY{\p_i}{\p_i}$.  As (i) $r = \cardin{R}$, (ii)
    $\centroX{R} = Y/\cardin{R} = Y/r$, (iii) $r \geq \zeta/\epsA^2$, and
    (iv) by Markov's inequality, we have
    \begin{align*}
        \Prob{\normX{\centroX{R}} > \epsA \avgp}
      &=%
        \Prob{\frac{\normX{Y}}{r} > \epsA \avgp}
        =%
        \Prob{\normX{Y}^2 > (r \epsA \avgp)^2}
        \leq%
        \frac{\Ex{\smash{\normX{Y}^2}} \bigr.}{(r \epsA \avgp)^2}
        \leq
        \frac{r \avgp^2}{(r \epsA \avgp)^2}
        =%
        \frac{1}{r \epsA^2 }
        \leq
        \frac{1}{\zeta}.
    \end{align*}
\end{proof}

\lemref{sample:mean} readily implies the no-dimensional \Caratheodory
theorem, see \apndref{caratheodory}.

\subsection{Approximate Tverberg theorem}

We now present the key technical lemma that will allow us to prove an
approximate Tverberg theorem.

\begin{lemma}
    \lemlab{main:a}%
    Let $\P$ be a set of $n$ points in $\Re^d$, and $\epsA \in (0,1)$
    be a parameter, and assume that $n \gg 1/\epsA^4$ (i.e.,
    $n \geq \cDelta/\epsA^4$, where $\cDelta$ is some absolute
    constant).  Let $\avgp = \avgpX{\P}$.  Then, one can compute, in
    $O(n d /\epsA^2)$ expected time, a partition of $\P$ into $k$ sets
    $\P_1, \ldots, \P_k$, and a ball $\ball$, such that
    \begin{compactenumi}
        \smallskip%
        \item $\forall i \ \cardin{\P_i} \leq 4/\epsA^2 + 9/2$,
        \smallskip%
        \item $\forall i \ \CHX{\P_i} \cap \ball \neq \emptyset$,
        \smallskip%
        \item $\radiusX{\ball} \leq \epsA \avgp$, and

        \smallskip%
        \item $ k \geq n/(4/\epsA^2+9/2)$.
    \end{compactenumi}
\end{lemma}
\begin{proof}
    Let $\ball = \ballY{\centroX{\P}}{\epsA \avgp}$.  Let
    $\zeta = 2(1+\epsA^2/8)$, and $M = \ceil{\smash{\zeta/\epsA^2}}$.
    We randomly partition the points of $\P$ into
    $t = \floor{n/M} > M$ sets $\Q_1,\ldots, \Q_t$, all of size either
    $M$ or $M+1$ (this can be done by randomly permuting the points of
    $\P$, and allocating each set a range of elements in this
    permutation). Thus, each $\Q_i$, for
    $i\in \IRX{t} = \{1,\ldots,t\}$, is a random sample according to
    \lemref{sample:mean} with parameter $\geq \zeta$. Thus, with
    probability $\geq 1-1/\zeta$, the set $\Q_i$, for $i \in \IRX{t}$,
    is $\epsA$-close -- that is,
    $\dY{\centroX{\Q_i}}{\centroX{\P}} \leq \epsA \avgp$, and $\Q_i$
    is then considered to be \emphw{good}.

    Let $Z$ be the number of bad sets in $\Q_1, \ldots,\Q_{t}$. The
    probability of a set to be bad is at most $1/\zeta$, and by
    linearity of expectations, $\Ex{Z} \leq t/\zeta$.  Let
    $\beta = t(1+\epsA^2/8)/\zeta = t/2$.  By Markov's inequality, we
    have
    \begin{equation}
        \Prob{Z \geq t/2}
        =%
        \Prob{Z \geq \beta}
        \leq%
        \frac{ \Ex{Z}}{\beta}
        \leq
        \frac{ t/\zeta}{ (1+\epsA^2/8)t /\zeta}
        =%
        \frac{1}{1+\epsA^2/8}%
        \leq%
        1-\frac{\epsA^2}{16}.
        \eqlab{key}%
    \end{equation}
    We consider a round of sampling \emph{successful} if
    $Z < \beta=t/2$. The algorithm can perform the random partition
    and compute the centroid for all $\P_i$ in $O(n d)$ time overall.
    Since a round is successful with probability $\geq \epsA^2/16$,
    after $\ceil{\smash{16/\epsA^2}}$ rounds, the algorithm succeeds
    with constant probability. This implies that the algorithm
    performs, in expectation, $O(1/\epsA^2)$ rounds till being
    successful, and the overall running time is $O( nd/\epsA^2)$ time
    in expectation.

    In the (first and final) successful round, the number of bad sets
    is $<t/2$ -- namely, it is strictly smaller than the number of
    good sets. Therefore, we can match each bad set $B$ in the
    partition to a unique good set $G$, and replace both of them by a
    new set $X = G \cup B$. That is, every good set absorbs at most
    one bad set, forming a new partition with roughly half the
    sets. For such a newly formed set $X$, we have that
    \begin{align*}
      \cardin{X}
      &=%
        \cardin{B} + \cardin{G}
        \leq%
        2(M+1)
        \leq
        2{\ceil{\smash{\zeta/\epsA^2}}} + 2
        =%
        2{\ceil{\frac{2(1+\epsA^2/8)}{\epsA^2}}}+2
      =%
      2\ceil{ \frac{2}{\epsA^2} + \frac{1}{4}}
      +2
      \leq%
      2\pth{\frac{2}{\epsA^2}+ \frac{5}{4}} + 2
      \leq%
      \frac{4}{\epsA^2} + \frac{9}{2}.
    \end{align*}
    The point $\centroX{G}$ is in $\CHX{G} \subset \CHX{X}$, and
    $\centroX{G}$ is in distance at most $\epsA \avgp$ from the
    centroid of $\P$. Thus, all the newly formed sets in the partition
    are in distance $\leq \epsA \avgp$ from $\centroX{\P}$, and
    $\CHX{X} \cap \ball \neq \varnothing$.

    Finally, we have that the number of sets in the merged partition
    is at least
    \begin{math}
        k \geq \frac{n}{4/\epsA^{2}+9/2}.
    \end{math}
\end{proof}

\begin{remark}
    The mysterious constant $c$ (from the inequality $n \geq c/\delta^4$
    of \lemref{main:a}) is used in the partition implicitly -- the
    number of sets in the partition needs to be even. Thus, one set
    might need to be absorbed in the other sets, or more precisely two
    sets, because of the rounding issues.  Namely, we first partition
    the set into groups of size $M$, and we need at least $2M+2$ sets
    in the partition to have size $M$ (one additional last set can
    have size smaller than $M$).  Thus, the proof requires that
    $n \geq (2M+2)M + M = (2M+3)M $. This is satisfied, for example,
    if $n \geq 27/\delta^4$.
\end{remark}

\begin{theorem}
    \thmlab{main:a}%
    Let $\P$ be a set of $n$ points in $\Re^d$, and
    $\epsA \in (0,1/\sqrt{2})$ be a parameter, and assume that
    $n \gg 1/\epsA^4$.  Then, one can compute, in $O(n d /\epsA^2)$
    expected time, a partition of $\P$ into sets $\P_1, \ldots, \P_k$,
    and a ball $\ball$, such that
    \begin{compactenumi}
        \smallskip%
        \item $\forall i \ \cardin{\P_i} \leq 4/\epsA^2 + 9/2$,
        \smallskip%
        \item $\forall i \ \CHX{\P_i} \cap \ball \neq \emptyset$,
        \smallskip%
        \item $\radiusX{\ball} \leq \epsA \ts \diamX{\P}$, and

        \smallskip%
        \item $ k \geq n/(4/\epsA^2 + 9/2)$.
    \end{compactenumi}
\end{theorem}
\begin{proof}
    Let $\diam = \diamX{\P}$.  Next we use the algorithm of
    \lemref{main:a} with parameter $\sqrt{2}\epsA$.  Observe that
    \begin{equation*}
        \radiusX{\ball}
        \leq%
        \sqrt{2}\epsA \avgp
        \leq%
        \sqrt{2}\epsA (\diam/\sqrt{2})
        =%
        \epsA \diam,
    \end{equation*}
    by \lemref{side:show}, where $\avgp = \avgpX{\P}$.

    Observe that the algorithm does not require the value of
    $\diamX{\P}$, but rather the value of $\avgpX{\P}$, which can be
    computed in $O(nd)$ time, see \Eqref{yo}.
\end{proof}

\begin{corollary}
    The expected running time of \thmref{main:a} can be improved to
    $O( nd)$, with two of the guarantees being weaker:
    \begin{compactenumI}
        \smallskip%
        \item The sets are bigger:
        $\forall i \ \cardin{\P_i} \leq 6/\epsA^2 + 7$.

        \smallskip%
        \item And there are fewer sets: $ k \geq n/(6/\epsA^2+7)$.
    \end{compactenumI}
\end{corollary}
\begin{proof}
    We use \lemref{main:a} as before, but now requiring only a third
    of the sets to be good, and merging triples of sets to get one
    final good set. Thus, the new sets are bigger by a factor of
    $3/2$, compared \thmref{main:a}, as they are now the result of
    merging three sets instead of two.

    The probability of success is now constant, as \Eqref{key} becomes
    \begin{equation*}
        \Prob{Z \geq \smash{\frac{2}{3}t}\Bigr.}
        =%
        \Prob{Z \geq \smash{\frac{4}{3} \cdot \frac{t}{2}} \Bigr.}%
        =%
        \Prob{Z \geq \smash{\frac{4}{3}}\beta\Bigr.}%
        \leq%
        \frac{ \Ex{Z}}{(4/3)\beta}
        \leq
        \frac{3}{4}.
    \end{equation*}
    Namely, the partition succeeds with probability at least $1/4$,
    which implies that the algorithm is done in expectation after a
    constant number of partition rounds.
\end{proof}

\begin{remark}
    The (existential) result of Adiprasito \etal \cite[Theorem
    1.3]{abmt-tchtw-20} has slightly worse constants, but it requires
    some effort to see, as they ``maximize'' the number of sets $k$
    (instead of minimizing the size of each set). Specifically, they
    show that one can partition $\P$ into $k$ sets, with the computed
    ball having radius $(2+\sqrt{2})\sqrt{{k}/{n}} \, \diamX{\P}$
    (intuitively, one wants $k$ to be as large as
    possible). Translating into our language, they require that
    \begin{equation*}
        (2+\sqrt{2})\sqrt{\frac{k}{n}}
        \leq \epsA
        \implies
        (2+\sqrt{2})^2\frac{k}{n}\leq \epsA^2
        \implies
        k\leq n \frac{\epsA^2}{(2+\sqrt{2})^2}.
    \end{equation*}
    Our result, on the other hand, states that $k$ is at least
    (over-simplifying for clarity) $n \tfrac{ \epsA^2}{2}$ (for
    $\epsA$ sufficiently small).  Adiprasito \etal mention, as a side
    note, that their constant improves to $1+\sqrt{2}$ under certain
    conditions. Even then, the constant in the above theorem is
    better.

    This improvement in the constant is small (and thus, arguably
    minor), but nevertheless, satisfying.
\end{remark}

\subsection{Tverberg halving}
\seclab{halving}

An alternative approach is to randomly halve the point set and observe
that the centroids of the two halves are close together. In this section,
we show that this line of thinking leads to various algorithms which can be
derandomized efficiently. Foundational to this approach is the
following lemma (which is a variant of \lemref{sample:mean}).
\begin{lemma}
    \lemlab{sample:mean:2}%
    Let $U = \{u_1, \ldots, u_{2n}\}$ be a set of $2n$ points in
    $\Re^d$ with $\diam = \diamX{U}$.  For $i=1,\ldots, n$, with
    probability $1/2$, let $p_i = u_{2i-1}, \q_i = u_{2i}$, or
    otherwise, let $p_i = u_{2i}, \q_i = u_{2i-1}$. Let
    $\P = \{ p_1,\ldots, p_n\}$ and $\Q = \{\q_1, \ldots, \q_n\}$. For
    any parameter $t \geq 1$, we have
    \begin{math}
        \Prob{ \dY{\centroX{\P}}{ \centroX{\Q}}%
           \geq%
           \frac{t}{\sqrt{n}} \diam } \leq \frac{1}{t^2}.
    \end{math}
\end{lemma}

\begin{proof}
    This follows by adapting the argument used in the proof of
    \lemref{sample:mean}, and the details are included here for the
    sake of completeness.

    Let $v_i = u_{2i-1} - u_{2i}$. Consider the random variable
    \begin{math}
        Y = \centroX{\P} - \centroX{\Q} =%
        \sum_{i=1}^n \frac{X_i v_i}{n},
    \end{math}
    where $X_i \in \{-1,+1\}$ is picked independently with probability
    half. We first observe that
    \begin{compactenumi}
        \smallskip%
        \item $\Ex{Y} = \sum_{i=1}^n \Ex{X_i}v_i/n = 0$, \smallskip%
        \item $\Ex{X_i^2} = 1$, and \smallskip%
        \item for $i < j$, $\Ex{X_i X_j} = 0$.
    \end{compactenumi}
    \smallskip%
    Thus, we have
    \begin{align}
      \Ex{\bigl.\smash{\normX{Y }^2}} = \Ex{\DotProdY{Y}{Y}}%
      &=
        \Ex{\Bigl.
        \smash{
        \DotProdY{\smash{\sum\nolimits_{i=1}^n \frac{X_iv_i}{n}}}%
        {\,\smash{\sum\nolimits_{i=1}^n \frac{X_i v_i}{n}}  \Bigr.}}}
    	\nonumber\\
      &=%
        \frac{1}{n^2} \sum_{i=1}^n \Ex{X_i^2} v_i^2
        +
        2 \frac{1}{n^2} \sum_{i < j} \Ex{X_i X_j v_i v_j}
        \nonumber\\
      &=%
        \frac{1}{n^2} \sum_{i=1}^n  v_i^2
        \leq
        \frac{n \diam^2}{n^2}
        =
        \frac{ \diam^2}{n},
        \eqlab{yey:2}%
    \end{align}
    since
    $\normX{v_i} = \dY{u_{2i-1}}{u_{2i}} \leq \diamX{U} = \diam$. By
    Markov's inequality, we have
    \begin{align*}
      \Prob{\normX{Y} > t \smash{\frac{\diam}{\sqrt{n}}}\Bigr.\ts}
      =
      \Prob{\Bigl. \smash{\normX{Y}^2 > t^2 \frac{\diam^2}{n}}}
      \leq%
      \frac{\Ex{\smash{\normX{Y}^2}\bigr.} }{t^2 \diam^2/n}
      \leq%
      \frac{1}{t^2}.
    \end{align*}
\end{proof}

\noindent%
\textbf{Remarks.}
\begin{compactenumA}
    \smallskip%
    \item \lemref{sample:mean:2} can be turned into an efficient
    algorithm using the same Markov's inequality argument used in
    \thmref{main:a}.  Specifically, for any parameter
    $\espB \in (0,1)$, one can compute a partition into two sets $\P$
    and $\Q$ with
    $\dY{\centroX{\P}}{ \centroX{\Q}} \leq (1+\espB)\diam/\sqrt{n}$,
    with probability of failure at most $\frac{1}{(1+\espB)^2}$.
    Thus, the probability of success in each round is at least
    \begin{equation*}
        1- \frac{1}{(1+\espB)^2}
        =%
        \frac{\espB(2+\espB)}{(1+\espB)^2}%
        \geq%
        \frac{2 \espB}{4}
        =
        \frac{ \espB}{2}.
    \end{equation*}
    Then, in expectation, the algorithm performs $O(1 /\espB)$ rounds,
    resulting in $O(n d /\espB)$ expected runtime overall.

    \smallskip%
    \item \lemref{sample:mean:2} implies that there exists a partition
    $\P$ and $\Q$ of $U$ such that
    \begin{equation*}
        \dY{\centroX{\P}}{ \centroX{\Q}} \leq \diam/\sqrt{n}.
    \end{equation*}
    The example is essentially tight, see \lemref{side:show}.

    \smallskip%
    \item As in the standard algorithm for computing a $\epsA$-net via
    discrepancy \cite{c-dmrc-01,m-gd-99}, one can apply repeated
    halving to get the desired Tverberg partition until the sets are
    the desired size. This provides a method for a
    deterministic algorithm, which we present in \secref{halving:det}.
\end{compactenumA}

\section{Applications}

\subsection{No-dimensional centerball}
\seclab{center:ball}

We present an efficient no-dimensional centerpoint theorem; the
previous version \cite[Theorem 7.1]{abmt-tchtw-20} did not present an
efficient algorithm.

\begin{corollary}[No-dimensional centerpoint]
    Let $\P$ be a set of $n$ points in $\Re^d$ and
    $\epsA \in (0, 1/2)$ be a parameter, where $n$ is sufficiently
    large (compared to $\epsA$). Then, one can compute, in
    $O(nd /\epsA^2)$ expected time, a ball $\ball$ of radius
    $\epsA \ts \diamX{\P}$, such that any halfspace containing $\ball$
    contains at least $\Omega(\epsA^2 n)$ points of $\P$.
\end{corollary}

\begin{proof}
    Follows by applying \thmref{main:a} and the observation that, for
    any halfspace containing the computed ball $\ball$, it must also
    contain at least one point from each set of the partition
    $\P_1, \dots, \P_k$, where $k = \Omega(\epsA^2 n)$. Thus, the ball
    $\ball$ is as desired.
\end{proof}

\subsection{No-dimensional weak \TPDF{$\eps$}{eps}-net theorem}
\seclab{weak}

Originally given by Adiprasito \etal \cite[Theorem
7.3]{abmt-tchtw-20}, we prove a version of the no-dimensional weak
$\eps$-net theorem with an improved dependence on the parameters.  For
a sequence $\Q = (q_1,\ldots, q_r ) \in \P^r$, let
$\centroX{\Q} = \sum_{i=1}^r q_i / r$.  We reprove
\lemref{sample:mean} under a slightly different sampling model.

\begin{lemma}
    \lemlab{sample:mean:w:r}%
    Let $\P$ be a set of $n$ points in $\Re^d$, and
    $\epsA \in (0,1/2)$ and $\zeta > 1$ be parameters.  Let
    $r \geq \zeta /\epsA^2$.  For a random sequence
    $\Q = (\q_1,\ldots, \q_r)$ picked uniformly at random from $\P^r$,
    we have that
    $\Prob{\dY{\centroX{\P}}{\centroX{\Q}} > \epsA \diam} \leq
    1/\zeta$, where $\diam = \diamX{\P}$.
\end{lemma}

\begin{proof}
    The argument predictably follows the proof of
    \lemref{sample:mean}, and the reader can safely skip reading it,
    as it adds little new. Assume that
    $\centroX{\P} = \sum_{i=1}^n \tfrac{1}{n} \p_i = 0$. Let
    $\avgp^2 = \sum_{i=1}^n \tfrac{1}{n}\p_i^2$ and
    $\Bigl.Y = \sum_{i=1}^r \q_i$. Then,
    \begin{math}
    	\Ex{Y}%
        =%
        \sum_{i=1}^r \Ex{\q_i} = 0.
    \end{math}
    As $\normX{Y }^2 = \DotProdY{Y}{Y}$, it follows that
    \begin{align*}
      \Ex{\bigl.\smash{\normX{Y }^2}}
      &=%
      \Ex{\Bigl.\smash{\bigl(\sum_{i=1}^r \q_i \bigr)^2}}
      =%
      \Bigl.\smash{\sum_{k=1}^r \Ex{\q_k^2} +
      2 \sum_{i<j} \Ex{ \smash{\q_i\q_j}}
      }\\
        &=%
        \sum_{k=1}^r
        \sum_{i=1}^n \tfrac{1}{n}\p_i^2
        +
        2 \sum_{i<j} \Ex{ \q_i}\Ex{\smash{\q_j}}
        =r \avgp^2.
    \end{align*}

    Since $\centroX{R} = Y/r$, $r \geq \zeta/\epsA^2$, and by Markov's
    inequality, we have
    \begin{align*}
      \Prob{\normX{\centroX{R}} > \epsA \avgp}
      &=%
        \Prob{\frac{\normX{Y}}{r} > \epsA \avgp}
        =%
        \Prob{\normX{Y}^2 > (r \epsA \avgp)^2}
        \leq%
        \frac{\Ex{\smash{\normX{Y}^2}} \bigr.}{(r \epsA \avgp)^2}
      \leq
      \frac{r \avgp^2}{(r \epsA \avgp)^2}
      =%
      \frac{1}{r \epsA^2}
      \leq
      \frac{1}{\zeta}.
    \end{align*}
\end{proof}

A sequence $\Q \in \P^r$ \emphi{collides} with a ball $\ball$ if
$\ball$ intersects $\CHX{\Q}$.  In particular, if
$\dY{\centroX{\P}}{\centroX{\Q}} \leq \epsA \diam$, then $\Q$ collides
with the ball $\ballY{\centroX{\P}}{\epsA \diam}$, where
$\diam = \diamX{\P}$.

\begin{lemma}[Selection lemma]
    \lemlab{selection}
    Let $\P$ be a set of $n$ points in $\Re^d$ and $\epsA \in (0,1)$
    be a parameter. Let $r = \ceil{\smash{2/\epsA^2}}$.  Then, the
    ball $\ball = \ballY{\centroX{\P}}{ \epsA \diam }$ collides with
    at least $ n^r/2$ sequences of $\P^r$.
\end{lemma}

\begin{proof}
    Taking $\zeta=2$, by \lemref{sample:mean:w:r}, a random
    $r$-sequence from $\P^r$ has probability at least half to collide
    with $\ball$, which readily implies that this property holds for
    half the sequences in $\P^r$.
\end{proof}

\begin{theorem}[No-dimensional weak \TPDF{$\eps$}{eps}-net]
    Let $\P$ be a set of $n$ points in $\Re^d$, with diameter $\diam$,
    and $\epsA, \eps \in (0,1)$ be parameters, where $2/\epsA^2$ is an
    integer.  Then, there exists a set $F \subset \Re^d$ of
    $\leq 2 \eps^{-2/\epsA^2}\Bigr.$ balls, each of radius
    $\epsA \ts \diam$, such that, for all $Y \subset \P$, with
    $\cardin{Y} \geq \eps n$, $F$ contains a ball of radius
    $\epsA \diam$ that intersects $\CHX{Y}$.
\end{theorem}
\begin{proof}
    Our argument follows Alon \etal \cite{abfk-pswec-92}.  Let
    $r = \smash{2/\epsA^2}$.  Initialize $F = \varnothing$, and
    let $\TSet = {\P}^r$. If there is a set $\Q \subset \P$, with
    $\cardin{\Q} \geq \eps n$, where no ball of $F$ intersects
    $\CHX{\Q}$, then applying \lemref{selection} to $\Q$, the
    algorithm computes a ball $\ball$, of radius $\epsA \diam$, that
    collides with at least $(\eps n)^r/2$ sequences of $\Q^r$.  The
    algorithm adds $\ball$ to the set $F$, and removes from $\TSet$
    all the sequences that collide with $\ball$. The algorithm
    continues till no such set $\Q$ exists.

    As initially $\cardin{\TSet} = n^r$, the number of iterations of
    the algorithm, and thus the size of $F$, is bounded by
    \begin{math}
        \frac{n^r}{(\eps n)^r/2} = 2/\eps^r.
    \end{math}
\end{proof}

\begin{remark}
    In the version given by Adiprasito \etal \cite[Theorem 7.3]{abmt-tchtw-20},
    the set $F$ has size at most $(2/\epsA^2)^{2/\epsA^2} \eps^{-2/\epsA^2}$,
    while our bound is $2 \eps^{-2/\epsA^2}$.
\end{remark}

\section{Derandomization}

\subsection{Derandomizing mean sampling}

\lemref{sample:mean} can be derandomized directly using conditional
expectations. We also present a more efficient derandomization scheme
using halving in \secref{derandomize:halving:scheme}.

\begin{lemma}
    \lemlab{s:m:d:slow}%
    Let $\P$ be a set of $n$ points in $\Re^d$. Then, for any integer
    $r \geq 1$, one can compute, in deterministic $O(d n^3)$ time, a
    subset $R \subset \P$ of size $r$, such that
    \begin{math}
        \dY{\centroX{\P}}{\centroX{R}} \leq \avgpX{\P}/\sqrt{r} \leq %
        \diamX{\P}/ \sqrt{2r},
    \end{math}
    where $\avgp= \avgpX{\P}$, see \Eqref{yo}.
\end{lemma}

\begin{proof}
    We derandomize the algorithm of \lemref{sample:mean}.  We assume
    for simplicity of exposition that $\centroX{\P} = 0$. Let $\Sample$ be a sample of size $r$ without replacement from $\P$, and let $I_i \in \{0,1\}$ be the indicator for the event that \(\p_i \in \Sample\).

    Let \(Y = \sum_{i=1}^n I_i \p_i\). Then, \(\centroX{\Sample} = Y / r\), and thus
    \begin{math}
        \dY{\centroX{\Sample}}{ \centroX{\P}} =%
        {\normX{Y}}/{r}.
    \end{math}
    Consider the quantity
    \begin{align*}
      \beta
      =%
      Z(x_1, \dots, x_t) = \ExCond{\smash{\normX{Y}^2}\bigr. }{  \Event},
      \qquad
      \Event \equiv
      \pth{I_1 = x_1, \dots, I_t = x_t},
    \end{align*}
    where the expectation is over the random choices of
    $I_{t+1},\ldots, I_n$. At the beginning of the $(t+1)$\th
    iteration, the values of $x_1, \ldots, x_{t}$ were determined in
    earlier iterations, and the task at hand is to decide what value
    to assign to $x_{t+1}$ that minimizes
    $Z(x_1, \dots, x_t,x_{t+1})$.  Thus, the algorithm computes
    $\beta_0 = Z(x_1, \dots, x_t,0)$ and
    $\beta_1 = Z(x_1, \dots, x_t,1)$.

    Using conditional expectations, \Eqref{y:squared} becomes
    \begin{equation}
        \beta%
        =%
        \ExCond{\smash{\normX{Y }^2}\bigl.}{\Event}
        =%
        \sum_{i=1}^{n} \ExCond{I_i}{\Event} \p_i^2 +
        2 \sum_{i<j} \ExCond{I_i I_j}{\Event} \p_i \p_j.
        \eqlab{compute:beta}
    \end{equation}
    Let $\alpha = \sum_{k=1}^t x_k$, and observe that $r-\alpha$
    points are left to be chosen to be in $\Sample$ after
    \(\Event\). As such, arguing as in \Eqref{pair:ex}, for $i < j$,
    we have
    \begin{equation}
        \ExCond{I_i}{\Event}
        =
        \begin{cases}
          x_i & i \leq t\\
          \frac{r-\alpha}{n-t} & i >t,
        \end{cases}
        \qquad\text{and}\qquad
    	\ExCond{I_i {}I_j}{\Event} %
        =%
        \begin{cases}
          x_i x_j
          &
            i < j \leq t
          \\
          x_i \frac{r-\alpha}{n-t}
          &
            i \leq t <j
          \\[0.1cm]%
          \frac{(r-\alpha)(r-\alpha-1)}{(n-t)(n-t-1)}
          &
            t <i < j.
        \end{cases}
        \eqlab{cond:expect:f}%
    \end{equation}
    This implies that the algorithm can compute $\beta$ in quadratic
    time directly via \Eqref{compute:beta}.  Similarly, the algorithm
    computes $\beta_0$ and $\beta_1$. Observe that
    \begin{equation*}
        \beta = Z(x_1, \dots, x_t)
        =
        \frac{r-\alpha}{n-t}
        \beta_1
        +
        \frac{n-t - (r -\alpha)}{n-t} \beta_0.
    \end{equation*}
    Namely, $\beta$ is a convex combination of $\beta_0$ and
    $\beta_1$. Thus, if $\beta_0 \leq \beta$ then the algorithm sets
    $x_{t+1}=0$, and otherwise the algorithm sets
    $x_{t+1}=1$.

    The algorithm now performs $n$ such assignment steps, for
    $t=0,\ldots, n-1$, to compute an assignment of $x_1, \ldots, x_n$
    such that $Z(x_1, \ldots, x_n) \leq \Ex{\smash{\normX{Y}^2}}$.
    Overall, this leads to a $O(dn^3)$ time algorithm. Specifically,
    the algorithm outputs a set $\Sample \subseteq \P$ of size $r$,
    such that
    \begin{math}
        \Sample = \Set{\p_i }{ x_i = 1, i=1,\ldots,n}.
    \end{math}
    Observe that
    $Z(x_1,\ldots, x_n) = \normX{ r \centroX{\Sample}}^2 \leq
    \Ex{\smash{\normX{Y}^2}}$. Thus, by \Eqref{y:squared} and
    \lemref{side:show}, we have
    \begin{equation*}
        \dY{\centroX{\Sample}}{\centroX{\P}}
        =%
        \normX{\centroX{\Sample}}
        \leq%
        \sqrt{\frac{\Ex{\smash{\normX{Y}^2}}\bigr.}{r^2}}
        \leq%
        \sqrt{\frac{r \avgp^2}{r^2}}
        =%
        \frac{\avgp}{\sqrt{r}}
        \leq
        \frac{\diamX{\P}}{\sqrt{2r}}.
    \end{equation*}
\end{proof}

With some care, the running time of the algorithm of
\lemref{s:m:d:slow} can be improved to $O( dn)$ time, but the details
are tedious, and we delegate the proof of the following lemma to
\apndref{derandomized_running_time}.

\SaveContent{\LemmaSMDBody}%
{ Let $\P$ be a set of $n$ points in $\Re^d$. Then, for any integer
   $r \geq 1$, one can compute, in $O(d n)$ deterministic time, a
   subset $R \subset \P$ of size $r$, such that
    \begin{math}
        \dY{\centroX{\P}}{\centroX{R}}%
        \leq%
        \avgpX{\P}/\sqrt{r}%
        \leq %
        \diamX{\P}/ \sqrt{2r}.
    \end{math}
 }

\begin{lemma}
    \lemlab{s:m:d:fast}%
    \LemmaSMDBody{}
\end{lemma}

\subsection{Derandomizing the halving scheme}
\seclab{derandomize:halving:scheme}

The algorithm of \lemref{sample:mean:2} can be similarly derandomized.

\begin{lemma}
    \lemlab{sample:mean:3}
    Let $U = \{u_1, \ldots, u_{2n}\}$ be a set of $2n$ points in
    $\Re^d$ with $\diam = \diamX{U}$. One can partition $U$, in
    deterministic $O( d n)$ time, into two equal size sets $\P$ and
    $\Q$, such that
    \begin{math}
        \dY{\centroX{\P}}{ \centroX{\Q}} \leq \diam / \sqrt{n}.
    \end{math}
\end{lemma}

\begin{proof}
    We follow \lemref{sample:mean:2}. To this
    end, let $v_i = u_{2i-1} - u_{2i}$, for $i=1,\ldots, n$. Let
    \begin{math}
        Y = \sum_{i=1}^n \frac{X_i v_i}{n},
    \end{math}
    where $X_i \in \{-1,+1\}$.  Next, consider the quantity
    \begin{equation*}
        Z(x_1, \dots, x_t) = \ExCond{\smash{\normX{Y}^2}\bigr.\! }{ \Event},
        \qquad
        \Event \equiv \pth{X_1 = x_1, \dots, X_t
           = x_t},
    \end{equation*}
    where the expectation is over the random choices of
    $X_{t+1},\ldots, X_n$.  By \Eqref{yey:2}, we have
    \begin{math}
        Z(x_1, \dots, x_t) =%
        \frac{1}{n^2} \sum_{i=1}^n v_i^2 + \frac{2}{n^2} \sum_{i < j}
        \ExCond{X_i X_j v_i v_j}{ \Event}.
    \end{math}
    The latter term is
    \begin{align*}
      \sum_{i < j}
      \ExCond{X_i X_j v_i v_j}{ \Event}
      &=%
        \sum_{i < j: i,j \leq t} x_i x_j v_i v_j
        +
        \sum_{i < j: i \leq t < j} \Ex{x_i X_j v_i v_j}
      +
      \sum_{i < j: t < i,j} \Ex{X_i X_j v_i v_j}
      \\&
      =%
      \sum_{i < j \leq t} x_i x_j v_i v_j,
    \end{align*}
    as $\Ex{X_i} = \Ex{X_iX_j} =0$.  Thus,
    \begin{math}
        Z(x_1, \dots, x_t) =%
        \frac{1}{n^2} \sum_{i=1}^n v_i^2 + \frac{2}{n^2}\sum_{i < j
           \leq t} x_i x_j v_i v_j.
    \end{math}
    The key observation is that
    \begin{equation*}
        Z(x_1, \ldots, x_{t})
        = \frac{Z(x_1, \ldots, x_{t},-1) +
           Z(x_1, \ldots,
           x_{t},+1)}{2}.
    \end{equation*}
    Our goal is to compute the assignment of $x_1, \ldots, x_n$
    that minimizes $Z$. Observe that
    \begin{equation*}
        D_t
        =%
        Z(x_1, \ldots, x_{t},+1) - Z(x_1, \ldots, x_{t})
        =%
        \frac{2}{n^2}\Bigl({\sum_{i <  t+1} x_i
           v_i} \Bigr) v_{t+1}.
    \end{equation*}
    If $D_t \leq 0$, then the algorithm sets $x_{t+1} = +1$, otherwise
    the algorithm sets $x_{t+1} =-1$.  The algorithm has to repeat
    this process for $t=1,\ldots, n$, and naively, each step takes
    $O( dn )$ time. Observe that if the algorithm maintains the
    quantity $V_t = \sum_{ i =1}^t x_i v_i$, then $D_t$ can be
    computed in $O(d)$ time. This determines the value of $x_{t+1}$,
    and the value of $V_{t+1} = V_t + x_{t+1} v_{t+1}$ can be
    maintained in $O(d)$ time. As each iteration takes $O(d)$ time,
    the algorithm overall takes $O(dn)$ time. By the end of this
    process, the algorithm have computed an assignment
    $x_1,\ldots, x_n$, with an associated partition of $U$ into $\P$
    and $\Q$.  By \Eqref{yey:2}, we have
    \begin{math}
        \displaystyle%
        \dY{\centroX{\P}}{ \centroX{\Q}}^2%
        \leq%
        \Ex{\bigl.\smash{\normX{Y }^2}} \leq { \diam^2}/{n}.
    \end{math}
\end{proof}

\subsection{A deterministic approximate Tverberg partition}
\seclab{halving:det}

\begin{lemma}
    \lemlab{main:b}
    Let $\P$ be a set of $n$ points in $\Re^d$, and
    $\epsA \in (0,1/4)$ be a parameter.  Then, one can compute, in
    $O(nd \log n)$ deterministic time, a partition of $\P$ into sets
    $\P_1, \ldots, \P_k$, and a ball $\ball$, such that
    \begin{compactenumi}
        \smallskip%
        \item $\forall i \ \cardin{\P_i} \leq 8/\epsA^2$, \smallskip%
        \item $\forall i \ \CHX{\P_i} \cap \ball \neq \emptyset$,
        \smallskip%
        \item $\radiusX{\ball} \leq \epsA \ts \diamX{\P}$, and
        \smallskip%
        \item $ k \geq n \epsA^2 / 8$.
    \end{compactenumi}
\end{lemma}

\begin{proof}
    Assume for the time being that $n$ is a power of $2$. As done for
    discrepancy, we halve the current point set, and then continue
    doing this recursively (on both resulting sets), using the algorithm of
    \lemref{sample:mean:3} at each stage.  Conceptually, this is done
    in a binary tree fashion, and doing this for $i$ levels breaks the
    point set into $2^i$ sets. Let $\ell_i$ be an upper bound on the
    distance of the centroid of a set in the $i$\th level from the
    centroid of its parent.  By \lemref{sample:mean:3}, we have
    $2 \ell_i \leq \diam /\sqrt{n/2^i}$ (where $i=1$ in the top
    level). Thus, repeating this process for $t$ levels, we have that
    the distance of any centroid at the leaves to the global centroid
    is bounded by
    \begin{align}
        L_t
        =%
        \sum_{i=1}^t \ell_i
        &\leq
        \sum_{i=1}^t \frac{\diam}{2\sqrt{n/2^i}}
        =%
        \frac{\diam}{\sqrt{2n}}\sum_{i=0}^{t-1} \sqrt{2^i}
        =
        \frac{\diam}{\sqrt{2n}} \pth{\frac{ 2^{t/2} - 1 }{ \sqrt{2} - 1}}\nonumber\\
        &\leq
        \frac{5 \diam}{2 \sqrt{2n}}  2^{t/2}
        =
        \frac{5 \diam}{2 \sqrt{2} } \sqrt{\frac{1}{n/2^t}}.
        \eqlab{geometric:sum}%
    \end{align}
    Solving for
    \begin{math}
        \frac{5 }{2 \sqrt{2} } \sqrt{\frac{1}{n/2^t}} \leq \delta,
    \end{math}
    we get that this holds for
    \begin{math}
        n/2^t \geq {3.2 }/{\delta^2}.
    \end{math}
    We stop our halving procedure once \(t\) is large enough such that
    the preceding inequality no longer holds, implying the stated
    bound on the size of each set.

    If $n$ is not a power of $2$ then we apply the above algorithm to
    the largest subset that has size that is a power of two, and then
    add the unused points in a round robin fashion to the sets
    computed.
\end{proof}

\begin{remark}
    \remlab{derandomize:2}%
    If instead of keeping both halves, as done by the algorithm of
    \lemref{main:b}, one throws one of the halves away, and repeats
    the halving process on the other half, we end up with a single
    sample.  One can repeat this halving process until the ``sample''
    size is $\Theta(1/\delta^2)$.  Using the same argument as in
    \Eqref{geometric:sum} to bound the error, we obtain a sample
    $\Sample$ of size $\Theta(1/\delta^2)$, such that
    $\dY{\centroX{\Sample}}{\centroX{\P}} \leq \epsA \,
    \diamX{\P}$. The running time is $\sum_i O(dn/2^i) =
    O(dn)$. Namely, we get a deterministic $O(dn)$ time algorithm that
    computes a sample with the same guarantees as \lemref{s:m:d:fast}
    -- this version is somewhat less flexible and the constants are
    somewhat worse.
\end{remark}

\section{Conclusions}
Given a data set, archetypal analysis
\cite{cmh-fraa-14}
aims to identify a small subset
of points such that all (or most) points in the data can be
represented as a sparse convex-combination of these
``archtypes''. Thus, for a sparse convex-combination of points, generating
a point can be viewed as an ``explanation'' of how it is being induced
by the data. It is thus natural to ask for as many \emph{independent}
explanations as possible for a point -- the more such combinations,
the more a point ``arises'' naturally from the data.  Thus, an
approximate Tverberg partition can be interpreted as stating that high
dimensional data has certain points (i.e., the centroid) that are
robustly generated by the data.

From a data-analysis point of view, an interesting open question is
whether one can do better than the ``generic'' guarantees provided
here. If, for example, a smaller radius centroid ball exists, can it be
approximated efficiently? Can a sparser convex-combination of points be
computed efficiently?

While these questions in the most general settings seem quite
challenging, even solving them in some special cases might be
interesting.

In addition, prior works consider other no-dimensional results, such
as a no-dimensional version of Helly's theorem \cite{abmt-tchtw-20},
and a no-dimensional version of the colorful Tverberg theorem
\cite{cm-ndtta-22}.  Our work did not address these problems because
of the focus on simplicity, and a possible further direction is to
address these variants with extensions of the techniques used here.

   \paragraph*{Acknowledgements}%

   The authors thank Ken Clarkson and Sandeep Sen for useful
   discussions. The authors thank the anonymous referees for their
   detailed comments.  Work by both authors was partially supported by
   NSF AF award CCF-2317241.

\BibTexMode{%
   \RegVer{%
      \bibliographystyle{alpha}%
   }%
   \bibliography{fat_separator}%
}%
\BibLatexMode{\printbibliography}

\appendix

\section{No-dimensional \Caratheodory theorem}
\apndlab{caratheodory}

The following result is well known \cite{m-ldg-02}, and even the
following proof is probably known. We include the proof because it is
quite elegant, and for the sake of completeness.

\begin{lemma}
    \lemlab{caratheodory}%
    Let $\P$ be a set of points in $\Re^d$, and let $\epsA$ be a
    parameter. Let $\p$ be any point in the convex-hull of $\P$. Then,
    there exists a point $\q \in \CHX{\P}$, such that $\q$ is the
    convex combination of at most $\ceil{\smash{1/(2\epsA^2)}}$ points
    of $\P$, such that $\dY{\p}{\q} \leq \epsA \diam$, where
    $\diam= \diamX{\P}$.
\end{lemma}
\begin{proof}
    Let $\P = \{\p_1, \ldots, \p_n\}$ and assume that
    $\p = \sum_{i=1}^n \alpha_i \p_i$. For simplicity, further assume that
    $\alpha_i = x_i/M$ for every \(i\), where the $x_i$'s and $M$ are
    integers. We also assume that $\p$ is at the origin (otherwise, translate $\P$
    so that it holds).

    Let $\Q$ be the multiset of cardinality $M$, where, for
    $i=1,\ldots, n$, a point $\p_i \in \P$ appears $x_i$
    times. Clearly, $\p$ is the centroid of $\Q$.  Picking a random
    sample $\Sample$ from $\Q$ of size $r$, and setting
    $Y = \sum_{\p \in \Sample} \p$, we have $\centroX{\Sample} = Y/r$.
    Since
    $0 \leq \Var{\normX{Y}} = \Ex{ \smash{\normX{Y}^2}} -
    \Ex{\normX{Y}}^2$, we have
    \begin{align*}
        \Ex{\dY{\p}{\centroX{\Sample}}}
        &=%
        \Ex{\normX{Y}/r}
        =
        \sqrt{
           \Ex{\normX{Y}}^2/r^2}\\
        &\leq%
        \sqrt{
           \Ex{\smash{\normX{Y}^2}}\bigr.}/r
        \leq
        \sqrt{
           r \avgp^2 \bigr.}/r
        =%
        \avgp/\sqrt{r}
        \leq
        {\diam}/{\sqrt{2r }}
        \leq
        \epsA \diam,
    \end{align*}
    by \Eqref{y:squared}, \lemref{side:show}, and setting
    $r \geq \ceil{\smash{1/(2\epsA^2)}}$.  In particular, there exists
    a choice of $\Sample$, such that $\dY{\p}{\centroX{\Sample}}$ is
    no larger than the expectation, which implies the claim.
\end{proof}

\section{Proof of \TPDF{\lemref{s:m:d:fast}}{Lemma 4.2}}
\apndlab{derandomized_running_time}

\RestatementOf{\lemref{s:m:d:fast}}{\LemmaSMDBody}

\begin{proof}
    We follow the argument of \lemref{s:m:d:slow}.  The algorithm
    computes the following prefix sums
    \begin{equation*}
        P_i =\sum_{k=1}^i p_k,
        \qquad%
        Q_i =\sum_{k=1}^i p_k^2,
        \qquad
        P^x_i =\sum_{k=1}^i x_k p_k,
        \qquad\text{and}\qquad%
        Q^x_i =\sum_{k=1}^ix_k p_k^2.
    \end{equation*}
    Specifically, in $O(dn)$ time, it precomputes $P_i$ and $Q_i$, for
    $i=1, \ldots, n$.  The algorithm also maintains the prefix
    sums $P^x_i$ and $Q^x_i$, for $i=1,\ldots, t-1$. As the next value
    of $x_{t}$ is determined, in the $t$\th iteration, the algorithm
    computes the two new prefix sums $P^x_{t}$ and $Q^x_{t}$ in
    $O(d)$ time.

    Recall that, in the $t$\th iteration, the algorithm has already
    determined the values of $x_1,\ldots, x_{t-1}$, and the value of
    $x_t$ is speculatively set by the algorithm (to either $0$ or $1$,
    this is done twice for both values), and it needs to compute the
    quantity
    \begin{equation*}
        \beta%
        =%
        Z(x_1, \dots, x_t) =%
        B + 2 C,
    \end{equation*}
    where
    \begin{equation*}
        B =
        \sum_{i=1}^{n} \ExCond{I_i}{\Event} \p_i^2
        \quad\text{and}\quad
        C = \sum_{i<j}
        \ExCond{I_i{} I_j}{\Event} \p_i \p_j,
    \end{equation*}
    see \Eqref{compute:beta}.  To this end, the algorithm first
    computes $P_t^x$ and $Q_t^x$ in $O(d)$ time.

    Now, recall that for $\alpha = \sum_{k=1}^t x_k$, we have
    \begin{math}
        \ExCond{I_i}{\Event}
        =
        \begin{cases}
          x_i & i \leq t\\
          \frac{r-\alpha}{n-t} & i >t,
        \end{cases},
    \end{math}
    see \Eqref{cond:expect:f}. This implies that
    \begin{equation*}
        B
        =%
        \sum_{i=1}^t x_i p_i^2 + \sum_{i=t+1}^n
        \frac{r-\alpha}{n-t}p_i^2
        =%
        Q_t^x
        +
        \frac{r-\alpha}{n-t} (Q_n - Q_t).
    \end{equation*}
    As such, $B$ can be computed in $O(d)$ time from the
    maintained prefix sums. \Eqref{cond:expect:f} also states that
    \begin{equation*}
    	\ExCond{I_i {}I_j}{\Event} %
        =%
        \begin{cases}
          x_i x_j
          &
            i < j \leq t
          \\
          x_i \frac{r-\alpha}{n-t}
          &
            i \leq t <j
          \\[0.1cm]%
          \frac{(r-\alpha)(r-\alpha-1)}{(n-t)(n-t-1)}
          &
            t <i < j,
        \end{cases}
    \end{equation*}
    which implies
    \begin{align*}
      C &=
          \underbrace{
          \sum_{i=1}^{t-1} x_i \p_i
          \sum_{j=i+1}^t
          x_j  \p_j
          }_{L_t}
          +
          \frac{r-\alpha}{n-t}
          \underbrace{\sum_{i=1}^{t} x_i p_i \sum_{j =t+1}^n  \p_j}_{M_t}
      +
      \frac{(r-\alpha)(r-\alpha-1)}{(n-t)(n-t-1)}
      \underbrace{
      \sum_{i=t+1}^{n-1} \p_i \sum_{j=i+1}^n
      \p_j}_{R_t}.
    \end{align*}

    Observe that
    \begin{align*}
      L_t
      &=%
        L_{t-1} + x_t p_t \sum_{i=1}^{t-2} x_i \p_i
        + x_{t-1}p_{t-1} x_t  \p_t
        =%
        L_{t-1} + x_t p_t P^x_{t-2}
        + x_{t-1}p_{t-1} x_t  \p_t.
    \end{align*}
    Namely, the quantity $L_{t}$ can be computed in $O(d)$ time (from
    $L_{t-1}$). Similarly, $M_t$ can be computed, in $O(d)$ time, as
    \begin{equation*}
      M_t
      =%
    	\sum_{i=1}^{t}
    	x_i \p_i
    	\sum_{j=t+1}^n
    	\p_j
        =%
    	P_t^x (P_n - P_t).
    \end{equation*}
    The final quantity to consider is
    \begin{align*}
      R_{t}
      &=%
        \sum_{i=t+1}^{n-1} \p_i ( P_n - P_i)
        =%
        p_{t+1} (P_n - P_{t+1})
        +
        \sum_{i=t+2}^{n-1} \p_i ( P_n - P_i)
        =
        p_{t+1} (P_n - P_{t+1}) + R_{t+1}.
    \end{align*}
    Namely, the algorithm can precompute the values of
    $R_n, R_{n-1}, \ldots, R_1$ in $O(d n)$ time overall. Thus,
    from the above quantities, \(C\) can be computed in \(O(d)\)
    time.

    We conclude that $\beta$ can be computed in $O(d)$ time in the
    $t$\th iteration. This calculation is done twice for $x_t=0$ and
    $x_t=1$, and the algorithm determines the value of $x_t$ by
    picking the value that minimizes $\beta$. Thus, doing this for
    \(t = 1, \ldots, n-1\), the overall running time of the algorithm
    is $O(d n)$.
\end{proof}

\end{document}